\documentclass[11pt,letterpaper]{article}
\usepackage{amsfonts}
\usepackage{amsmath,amssymb,color}
\usepackage{graphicx}
\usepackage[singlespacing]{setspace}

\usepackage{natbib}
\newcommand{\jc}[1]{{{\small \color{red}\sc  (#1)}}}
\def\AA{\mbox{$\mathbf A$}}

\newcommand{\EE}{\mbox{$\mathbf E$}}
\newcommand{\ee}{\mbox{$\mathbf e$}}

\newcommand{\pp}{\mbox{$\mathbf p$}}

\newcommand{\RR}{\mbox{$\mathbf R$}}

\newtheorem{theorem}{Theorem}

\newtheorem{definition}{Definition}

\newtheorem{lemma}{Lemma}

\newtheorem{proposition}{Proposition}
\newtheorem{remark}{Remark}

\newenvironment{proof}[1][Proof]{\noindent\textbf{#1.} }{\ \rule{0.5em}{0.5em}}

\begin{document}

\title{Learning Item-Attribute Relationship in $Q$-Matrix Based Diagnostic Classification Models}

\author{Jingchen Liu, Gongjun Xu, and Zhiliang Ying\\\\ Columbia University}

\maketitle

\begin{abstract}
Recent surge of interests in cognitive assessment has led to the
developments of novel statistical models for diagnostic
classification. Central to many such models is the well-known
$Q$-matrix, which specifies the item-attribute relationship. This
paper proposes a principled estimation procedure for the
$Q$-matrix and related model parameters. Desirable theoretic
properties are established through large sample analysis. The
proposed method also provides a platform under which important
statistical issues, such as hypothesis testing and model
selection, can be addressed.

\noindent{\it Keywords:} Cognitive
assessment, consistency, DINA model, DINO model, latent traits,
model selection, optimization, self-learning,
statistical estimation.
\end{abstract}


\section{Introduction}
Diagnostic classification models (DCM) are important statistical tools in cognitive diagnosis and have widespread applications in educational measurement, psychiatric evaluation, human resource development, and many other areas in science, medicine, and business. A key component in many such models is the so-called $Q$-matrix, first introduced by \cite{Tatsuoka1983}; see also \cite{Tatsuoka} for a detailed coverage. The $Q$-matrix specifies the item-attribute relationship, so that responses to items can reveal attributes configuration of the respondent. In fact, \cite{Tatsuoka1983, Tatsuoka} proposed the rule space method that is simple and easy-to-use.

Flexible and sophisticated statistical models can be built around the $Q$-matrix. Two such models are the DINA model  \citep[Deterministic Input, Noisy Output ``AND'' gate; see][]{Junker} and the DINO model \citep[Deterministic Input, Noisy Output ``OR'' gate; see][]{Templin2006, Templin}. Other important developments can be found in \cite*{Tatsuoka1985,DiBello,Junker, Hartz,TatsuokaC,AHM,Templin2006,Chiu}.  \cite*{Rupp} contains a comprehensive summary of many classical and recent developments.

There is a growing literature on the statistical inference of
$Q$-matrix based DCMs that addresses the issues of estimating item
parameters when the $Q$-matrix is prespecified
\citep*{Rupp2002,Henson, RoussosTH, Stout2007}. Having a correctly
specified $Q$-matrix is crucial both for parameter estimation
(such as the slipping, guessing probability, and the attribute
distribution) and for the identification of subjects' underlying
attributes. As a result, these approaches are sensitive to the
choice of the $Q$-matrix  \citep{Rupp20082, dela2, dela}. For
instance, a misspecified $Q$-matrix may lead to substantial lack
of fit and, consequently, erroneous attribute identification.
Thus, it is desirable to be able to detect misspecification and to
obtain a data driven $Q$-matrix.

In contrast, there has not been much work about estimation of the
$Q$-matrix. To our knowledge, the only rigorous treatment of the
subject  is given by  \cite*{LXY2011}, which defines an estimator
of the $Q$-matrix under the DINA model assumption and provides
regularity conditions under which desirable  theoretical
properties are established. The work of this paper may be viewed
as the continuation of \cite{LXY2011} in the sense that it
completes the estimation of the $Q$-matrix for the DINA model and
extends the estimation procedure (as well as the consistency
results) to the DINO model. The DINA and the DINO models impose
rather different interactions among attributes. However, we show
that there exists a duality between the two models. This
particular feature is interesting especially for theoretical
development, as it allows us to adapt the results and analysis
techniques developed for the DINA model to the DINO model without
much additional effort. This will be shown in our technical
developments.

The main contribution of this paper is two-fold. First, it
provides a rigorous analysis of the $Q$-matrix for the DINA model
when both the slipping and guessing parameters are unknown. This
is a substantial extension of the results in  \cite{LXY2011} which
requires a complete knowledge of the guessing parameter. It gives
a definitive answer to the estimability of the $Q$-matrix for the
DINA model by presenting a set of sufficient conditions under
which a consistent estimator exists. Second, we conduct a parallel
analysis (to the analysis for the DINA model) for the DINO model.
In particular, a consistent estimator of the $Q$-matrix for the
DINO model and its properties are presented. Thanks to the duality
structure, part of the intermediate results developed for the DINA
model can be borrowed to the analysis of the DINO model.

One may notice that our estimation procedure is in fact generic in
the sense that it is implementable to a large class of DCMs
besides the DINA and DINO models.  In particular, the procedure is
implementable to the NIDA (Noisy Inputs, Deterministic ``And''
Gate) model and the NIDO (Noisy Inputs, Deterministic ``Or'' Gate)
model among others, though theoretical properties under such model
specifications still need to be established. In addition to the
estimation of the $Q$-matrix, we emphasize that the idea behind
the derivations forms a principled inference framework. For
instance, during the course of the description of the estimation
procedure, necessary conditions for a correctly specified
$Q$-matrix are naturally derived. Such conditions can be used to
form appropriate statistics for hypothesis testing and model
diagnostics. In that connection, additional developments (e.g. the
asymptotic distributions of those statistics) are needed, but they
are not the focus of the current paper. Therefore, the proposed
framework can potentially serve as a principled inference tool for
the $Q$-matrix in diagnostic classification models.

This paper is organized as follows. Section \ref{SecMain} contains
the main ingredient: presentation of the estimation procedures for
both the DINA and DINO models and the statement of the consistency
results. Section \ref{SecDisc} includes further discussions of the
theorems and various issues. The proofs of the main theorems in
Section \ref{SecMain} and several important propositions are given
in Section \ref{SecProof}. The most technical proofs of two
central propositions are given in the Appendix.

\section{Main results}\label{SecMain}

\subsection{Notation and model specification}

The specification of the diagnostic classification models
considered in this paper consists of the following concepts.

\emph{Attribute:} subject's underlying mastery of certain skills or presence of certain mental health conditions. There are $k$ attributes and we use $\AA= (A^1,...,A^k)^\top$ to denote the vector of attributes, where $A^j=1$ or $0$, indicating presence or absence of the $j$-th attribute, $j=1, \dots, k$.

\emph{Responses to items:} There are $m$ items and we use $\RR = (R^1,...,R^m)$ to denote the vector of responses to them. For simplicity, we assume that $R^j\in \{0,1\}$ is a binary variable for each $j=1,...,m$.

Note that both $\AA$ and $\RR$ are subject specific. Throughout this paper, we assume that the number of attributes $k$ is known and that the number of items $m$ is always observed.

\emph{$Q$-matrix:} the link between the items and the attributes.
In particular, $Q=(Q_{ij})_{m\times k}$ is an $m\times k$ matrix
with binary entries. For each $i$ and $j$, $Q_{ij}=1$ indicates
that  item $i$ requires attribute $j$ and $Q_{ij}=0$ otherwise.

We define \emph{capability indicator}, $\xi(\AA,Q)$, which
indicates if a subject possessing attribute profile $\AA$ is
capable of providing a positive response to item $i$ if the
item-attribute relationship is specified by matrix $Q$. Different
capability indicators give rise to different DCMs. For instance,
\begin{equation}\label{xidina}
\xi^i_{DINA}(\AA,Q)= \mathbf 1 (\mbox{$A^j \geq Q_{ij}$ for all $j=1,...,k$})
\end{equation}
is associated with the DINA model, where $\mathbf 1$ is the usual indicator function.
The DINA model assumes conjunctive relationship among attributes, that is, it is necessary to possess all the attributes indicated by the $Q$-matrix to be capable of providing  a positive response to an item. In addition, having additional unnecessary attributes does not compensate for the lack of the necessary attributes. The DINA model is particularly popular in the context of educational testing.

Alternative to the ``and'' relationship, one may impose an ``or'' relationship  among the attributes, resulting in the DINO model. The corresponding capability indicator takes  the following form
\begin{equation}\label{xidino}
\xi^i_{DINO}(\AA,Q)= \mathbf 1 (\mbox{there exists a $j$ such that $A^j \geq Q_{ij}$}).
\end{equation}
That is, one needs to possess at least one of the required
attributes to be capable of responding positively to that item.

The last ingredient of the model specification is related to the so-called slipping and guessing parameters. The names ``slipping'' and ``guessing'' arise from the educational applications. The slipping parameter is the probability that a subject (with attribute profile $\AA$) responds negatively to an item if the capability indicator to that item $\xi_{DINA}(\AA,Q) =1$; similarly, the guessing parameter refers to the probability that a subject's responds positively if his/her capability indicator $\xi_{DINA}(\AA,Q)=0$. We use $s$ to denote the slipping probability and $g$ to denote the guessing probability (with corresponding subscript indicating different items). In the technical development, it is more convenient to work with the complement of the slipping parameter. Therefore, we define $c= 1- s$ to be the correctly answering probability, with $s_i$ and $c_i$ being the corresponding item-specific notation. Given a specific subject's profile $\AA$, the response to item $i$ under the DINA model follows a Bernoulli distribution
\begin{equation} \label{prob}P(R^i = 1|\AA) = c_i^{\xi^i_{DINA}(\AA,Q)}g_i^{1-\xi^i_{DINA}(\AA,Q)}.\end{equation}
With the same definition of $c_i$ and $g_i$, the response under the DINO model follows
\begin{equation} \label{probdino}P(R^i = 1|
\AA) = c_i^{\xi^i_{DINO}(\AA,Q)}g_i^{1-\xi^i_{DINO}(\AA,Q)}.\end{equation}
In addition, conditional on $\AA$, $(R^1,...,R^m)$ are jointly independent.

Lastly, we use subscripts to indicate different subjects. For instance, $\RR_r= (R^1_r,...,R^m_r)^\top$ is the response vector of subject $r$. Similarly, $\AA_r$ is the attribute vector of subject $r$. With $N$ subjects, we observe $\RR_1,...,\RR_N$ but not $\AA_1, ..., \AA_N$.  Thus, we finished our model specification.

\subsection{Estimation of the $Q$-matrix}

In this section, we develop a general approach to the estimation of the $Q$-matrix and item parameters. We first deal with the DINA model and then, via introducing a duality relation, the DINO model.

\subsubsection{DINA model}

We need to introduce additional notation and concepts. Throughout
the discussion, we use $Q$ to denote the true matrix and $Q'$ to
denote a generic $m\times k $ binary matrix.

\emph{Attribute distribution.} We assume that the subjects are a random sample (of size $N$) from a designated population so that their attribute profiles, $\AA_r$, $r=1,..., N$ are i.i.d. random variables, with the following distribution
\begin{equation}\label{AttPop}P(\AA_r = \AA)= p_{\AA},\end{equation}
where, for each $\AA\in\{0,1\}^{k}$, $p_{\AA}\in [0,1]$ and $\sum_{\AA}p_{\AA}=1$.
We use $\pp=(p_{\AA}: \AA\in \{0,1\}^k)$ to denote the distribution  of the attribute profiles.

\emph{The $T$-matrix.} The $T$-matrix is a non-linear function of the $Q$-matrix and provides a linear relationship between the attribute distribution and the response distribution. In particular, let $T(Q)$ be a matrix of $2^k$ columns. Each column of $T$ corresponds to one attribute profile $\AA\in \{0,1\}^k$.  To facilitate the description, we use binary vectors of length $k$ to label the the columns of $T(Q)$ instead of using ordinal numbers. For instance, the $\AA$-th column of $T(Q)$ is the column that corresponds to attribute $\AA$.

Let $I_i$ be a generic notation for a positive response to item
$i$. Let ``$\wedge$'' stand for ``and'' combination. For instance,
$I_{i_1}\wedge I_{i_2}$ denotes  positive responses to both item
$i_1$ and $i_2$. Each row of $T(Q)$ corresponds to one item or one
``and'' combination of items, for instance, $I_{i_1}$,
$I_{i_1}\wedge I_{i_2}$, or $I_{i_1}\wedge I_{i_2}\wedge
I_{i_3}$,... For $T(Q)$ containing all the single items and all
``and'' combinations, it has $2^m-1$ rows. We will later say that
such a $T(Q)$ is \emph{saturated}.

We now proceed to the description of each row vector of $T(Q)$. We
define $B_Q(I_i)$ to be a $2^k$ dimensional row vector. Using the
same labeling system as that of the columns of $T(Q)$, the
$\AA$-th element of $B_Q(I_i)$ is defined as
$\xi^i_{DINA}(\AA,Q)$, that is, this element indicates if a
subject with attribute $\AA$ is capable of responding positively
to item $i$. Thus, $B_{Q}(I_{i})$ is the vector indicating the
attribute profiles that is capable of responding positively to
item $i$.

Using a similar notation, we define that
\begin{equation}
B_{Q}(I_{i_{1}}\wedge ...\wedge I_{i_{l}})=\Upsilon
_{h=1}^{l}B_{Q}(I_{i_{h}}),  \label{ProdB}
\end{equation}
where the operator ``$\Upsilon _{h=1}^{l}$'' is element-by-element multiplication from $B_Q(I_{i_1})$ to $B_Q(I_{i_l})$. For instance,
$$W=\Upsilon_{h=1}^l V_h$$
means that $W^j = \prod_{h=1}^l V_{h}^j$, where $W=(W^1,...,W^{2^k-1})$ and $V_h=(V_h^1,...,V_h^{2^k-1})$. Therefore, $B_Q(I_{i_1}\wedge ... \wedge I_{i_l})$ is the vector indicating the attributes that are capable of responding positively to items $i_1,...,i_l$. The row in $T(Q)$ corresponding to $I_{i_{1}}\wedge ...\wedge I_{i_{l}}$ is $B_{Q}(I_{i_{1}}\wedge...\wedge I_{i_{l}})$.

\emph{$\alpha$-vector.} We let $\alpha $ be a column vector whose
length is equal to the number of rows in $T(Q)$. Each component in
$\alpha$ corresponds to a row vector of $T(Q)$. The element in
$\alpha$ corresponding to $I_{i_{1}}\wedge ...\wedge I_{i_{l}}$ is
$N_{I_{i_{1}}\wedge...\wedge I_{i_{l}}}/N$, where
$N_{I_{i_{1}}\wedge...\wedge I_{i_{l}}}$ denotes the number of
people with positive responses to items $i_1,...,i_l$, that is
$$N_{I_{i_{1}}\wedge ...\wedge I_{i_{l}}}= \sum_{r=1}^N \prod_{j=1}^lR_r^{i_j}.$$

\paragraph{No slipping or guessing.} We first consider a simplified situation in
which both the slipping and guessing probabilities are zero.  Under this special situation, \eqref{prob} implies that
$$R^i_r = \xi^i_{DINA} (\AA_r), \quad i=1,...,m; \; r=1,..., N. $$
In other words, the probabilistic relationship becomes a certainty
relationship. We further let $\hat {\pp}=\{\hat p_{\AA}: \AA\in
\{0,1\}^k\}$ be the (unobserved) empirical distribution of the
attribute profiles, that is,
$$\hat p_{\AA} = \frac 1 N \sum_{r=1}^N \mathbf 1 (\AA_r = \AA).$$
Note that each row vector of $T(Q)$ indicates the attribute profiles that are capable of responding positively to the corresponding item(s). Then, for each set of $i_1,...,i_l$, we may expect the following identity
$$\frac{N_{i_1\wedge...\wedge i_l}}N =  B_Q(I_{i_1}\wedge .... \wedge I_{i_l}) \hat {\pp},$$
where $B_Q$ is a row vector and $\hat {\pp}$ is a column vector. Therefore, thanks to the construction of $T(Q)$ and vector $\alpha$, in absence of possibility of slipping and guessing, we may expect the following set of linear equations holds
$$T(Q)\hat {\pp} = \alpha.$$
Note that $\hat {\pp}$ is not observed. The above display implies that if the $Q$-matrix is correctly specified and the slipping and guessing probabilities are zero, then the linear equation $T(Q)\pp = \alpha$ (with $\pp$ being the variable) has at least one solution. For each binary matrix $Q'$, we define that
$$S(Q')= \inf_{\pp} |T(Q') \pp- \alpha|,$$
where the minimization is subject to the constraints that
$p_{\AA}\in [0,1]$ and $\sum_{\AA}p_{\AA}=1$. Based on the above
results, we may expect that $S(Q)=0$ and therefore $Q$ is one of
the minimizers of $S(Q)$. In addition, the empirical distribution
$\hat {\pp}$ is one of the minimizers of $|T(Q) \pp- \alpha|$.
Therefore, we just derived a set of necessary conditions for a
correctly specified $Q$-matrix. In our subsequent theoretical
developments, we will show that under some circumstances these
conditions are also sufficient.

\paragraph{Illustrative example.} To aid the understanding of the $T$-matrix, we provide one simple example. Consider the
following $3\times 2$ $Q$-matrix,
\begin{equation}\label{Q}
Q=\quad
\begin{tabular}{ccc}
\hline\hline & addition & multiplication\\ \hline
$2+3$ & $1$ & $0$ \\
$5\times 2$ & $0$ & $1$ \\
$(2+3)\times 2$ & $1$ & $1$\\\hline
\end{tabular}
\end{equation}
and the contingency table of attributes
\begin{center}
\begin{tabular}{cc}
\hline\hline
& multiplication \\
addition &
\begin{tabular}{cc}
$\hat p_{00}$ & $\hat p_{01}$ \\
$\hat p_{10}$ & $\hat p_{11}$
\end{tabular}
\\ \hline
\end{tabular}
\end{center}
Note that if the $Q$-matrix is correctly specified and the
slipping and guessing probabilities are all zero we should be able
to obtain the following identities
\begin{equation}
N(\hat p_{10}+\hat p_{11})=N_{I_{1}},\quad N(\hat p_{01}+\hat
p_{11})=N_{I_{2}},\quad N\hat p_{11}=N_{I_{3}}.  \label{margin}
\end{equation}%
We then create the corresponding $T$-matrix and $\alpha$-vector as follows
\begin{equation}
T(Q)=\left(
\begin{array}{cccc}
0& 1 & 0 & 1 \\
0& 0 & 1 & 1 \\
0& 0 & 0 & 1
\end{array}%
\right) ,\quad \alpha =\left(
\begin{array}{c}
N_{I_{1}}/N \\
N_{I_{2}}/N \\
N_{I_{3}}/N
\end{array}%
\right) .  \label{Num}
\end{equation}%
The first column of $T(Q)$ corresponds to the zero attribute profile; the second corresponds to $\AA = (1,0)$; the third corresponds to $\AA=(0,1)$; and the last corresponds to $\AA= (1,1)$. The first row of $T(Q)$ corresponds to item $2+3$, the second to $5\times 2$, the third to $(2+3)\times 2$. In addition, we may further consider combinations such as
$$N\hat p_{11} = N_{I_{1}\wedge I_{2}}.$$
The corresponding $T$-matrix and $\alpha$-vector should be
\begin{equation}\label{NewT}
T(Q)=\left(
\begin{array}{cccc}
0 & 1 & 0 & 1 \\
0 & 0 & 1 & 1 \\
0 & 0 & 0 & 1 \\
0 & 0 & 0 & 1%
\end{array}%
\right) ,\quad \alpha =\left(
\begin{array}{c}
N_{I_{1}}/N \\
N_{I_{2}}/N \\
N_{I_{3}}/N \\
N_{I_{1}\wedge I_{2}}/N%
\end{array}%
\right) .
\end{equation}%
Under the DINA model assumption and $g_{i}=s_{i}= 1- c_{i}= 0$, we obtain that
$$T(Q)\hat {\pp} = \alpha.$$

\paragraph{Nonzero slipping and guessing probabilities.} We next extend the necessary conditions just derived to nonzero but known
slipping and guessing probabilities. To do so, we need to modify
the $T$-matrix. Let $T_{c,g}(Q)$ be a matrix with the same
dimension as that of $T(Q)$, with each row vector being defined
slightly differently to incorporate the slipping and guessing
probability. In particular, let
$$B_{c,g,Q}(I_i) = (c_i -g_i) B_Q (I_i) + g_i \EE$$
where $\EE= (1,...,1)$ is the row vector of ones and $c_i$ is the
positive responding probability of item $i$. In addition, we let
\begin{equation}
B_{c,g,Q}(I_{i_{1}}\wedge ...\wedge I_{i_{l}})=\Upsilon
_{h=1}^{l}B_{c,g,Q}(I_{i_{h}}),  \label{ProdBcg}
\end{equation}
Clearly, each element of $B_{c,g,Q}(I_i)$ is the probability of
observing a positive response to item $i$ for a certain attribute
profile. Likewise, elements of $B_{c,g,Q}(I_{i_{1}}\wedge
...\wedge I_{i_{l}})$ indicate the probabilities of positive
responses to items $i_1$,..., $i_l$. The row in $T_{c,g}(Q)$
corresponding to $I_{i_{1}}\wedge ...\wedge I_{i_{l}}$ is
$B_{c,g,Q}(I_{i_{1}}\wedge...\wedge I_{i_{l}})$. To facilitate our
statement, we define that
\begin{equation}\label{Tc}
T_{c}(Q) = T_{c,\mathbf 0}(Q),
\end{equation}
where $\mathbf 0 = (0,...,0)^\top$ is the zero vector. That is,
$T_c(Q)$ is the matrix $T_{c,g}(Q)$ with guessing probabilities
being zero.

Recall that $\pp$ is the attribute distribution. Thus,
\begin{equation*}P(R^{i_1}=1,...,R^{i_l}=1)= E(P(R^{i_1}=1,...,R^{i_l}=1|\AA)) = B_{c,g,Q}(I_{i_{1}}\wedge ...\wedge I_{i_{l}}) \pp. \end{equation*}
Further, we obtain that
\begin{equation*} E(\alpha) = T_{c,g}(Q)\pp.\end{equation*}
In presence of slipping and guessing, one cannot expect to solve
equation $T_{c,g}(Q)\pp = \alpha$ exactly the same way as in the
case of no guessing and slipping. On the other hand, thanks to the
law of large numbers, we obtain that $\alpha \rightarrow
E(\alpha)$ as $N\rightarrow \infty$. Then this equation can be
solved asymptotically. Thus, for a generic $Q'$, we defined the
loss function
\begin{equation}\label{scorecg}S_{c,g}(Q') = \inf_{\pp} |T_{c,g}(Q') \pp- \alpha|,\end{equation}
where the above optimization is subject to the constraint that
$p_{\AA}\in [0,1]$ and $\sum_{\AA} p_{\AA} =1$ and $|\cdot|$ is
the Euclidean normal. In view of the preceding argument, we expect
that
\begin{equation}
S_{c,g}(Q)\rightarrow 0
\end{equation}
almost surely as $N\rightarrow \infty$, that is, the true
$Q$-matrix asymptotically minimizes the criterion function
$S_{c,g}$. This leads us to propose the following estimator of $Q$
\begin{equation}\label{est}
\hat Q(c,g) =\arg \inf_{Q'} S_{c,g}(Q'),
\end{equation}
where $(c,g)$ is included in $\hat Q$ to indicate that the
resulting estimator requires the knowledge of the correct
responding and guessing probabilities.

\paragraph{Situations when $c$ and $g$ are unknown.}
Suppose that for a given $Q'$, we can construct an estimator
$(\hat c(Q'),\hat g(Q'))$ of $(c,g)$. In addition, suppose that
$(\hat c(Q),\hat g(Q))$ is consistent, that is, $(\hat c(Q), \hat
g(Q)) \rightarrow (c,g)$ in probability as $N\rightarrow \infty$.
Then, we define
\begin{equation}\label{estc}
\hat Q_{\hat c, \hat g} = \arg\inf_{Q'} S_{\hat c(Q'),\hat g(Q')}(Q'),
\end{equation}
that is, we plug in the estimator of $(c,g)$ into the objective function in \eqref{est}. We will present one specific choice of $(\hat c, \hat g)$ in Section \ref{SecPara}.
%

\subsubsection{DINO model}

We now proceed to the description of the estimation procedure of the DINO model. The DINO can be considered as the dual model of the DINA model. The estimation procedure is similar except that the ``AND'' relationship needs to be changed to an ``OR'' relationship. In subsequent technical development, we will provide the precise meaning of the duality. First, we present the construction of the estimator.

\emph{The $U$-matrix.} The matrix $U_{c,g}(Q)$ is similar to
$T_{c,g}(Q)$ except that it admits an ``OR'' relationship among
items. In particular, first define $F_Q(I_i)$ to be a vector of
$2^k$ dimension and the $\AA$-th element is defined as
$\xi_{DINO}(\AA,Q)$. Therefore, $F_Q(I_i)$ indicates the attribute
profiles that are capable of providing positive responses to item
$i$. We use ``$\vee$'' to denote the ``OR'' combinations among
items and define
$$F_Q(I_{i_1}\vee...\vee I_{i_l}) =\EE - \Upsilon_{j=1}^l (\EE - F_Q(I_{i_j})).$$
Thus, $F_{Q}(I_{i_1}\vee...\vee I_{i_l})$ is a vector indicating
the attribute profiles that are capable of responding positively
to at least one of the item(s) $i_1$,..., $i_l$. We let the row in
$U(Q)$ corresponding to $I_{i_1}\vee...\vee I_{i_l}$ be
$F_Q(I_{i_1}\vee...\vee I_{i_l})$. In presence of slipping and
guessing, we define
$$F_{c,g,Q}(I_i) = (c_i - g_i ) F_{Q}(I_i) + g_i \EE$$
and
$$F_{c,g,Q}(I_{i_1}\vee...\vee I_{i_l}) =\EE - \Upsilon_{j=1}^l (\EE - F_{c,g,Q}(I_{i_j})).$$
We let the row in $U_{c,g}(Q)$ corresponding to ``$I_{i_1}\vee...\vee I_{i_l}$'' be $F_{c,g,Q}(I_{i_1}\vee...\vee I_{i_l})$.

\emph{The $\beta$-vector.} The vector $\beta$ plays a similar role
as the vector $\alpha$ for the DINA model. Specifically, $\beta$
is a column vector whose length is equal to the number of rows of
$U(Q)$. Each element of $\beta$ corresponds to one row vector of
$U(Q)$. The element of $\beta$ corresponding to
$I_{i_1}\vee...\vee I_{i_l}$ is defined as
$$N_{I_{i_1}\vee ... \vee I_{i_{l}}} /N= \frac 1 N\sum_{r=1}^N \mathbf 1(\mbox{there exists a $j$ such that $R_r^{i_j}=1$}).$$

With such a construction and a correctly specified $Q$, one may expect that
$$\beta \rightarrow U_{c,g}(Q) \pp$$
almost surely as $N\rightarrow \infty$. Therefore, we define
objective function
\begin{equation}\label{ScoreDINO}
V_{c,g}(Q) = \inf_{\pp} |U_{c,g}(Q) \pp - \beta|,
\end{equation}
where $\inf$ subject to $\sum_{\AA} p_{\AA} = 1$ and $p_{\AA}\in
[0,1]$. Furthermore, an estimator of $Q$ can be obtain by
\begin{equation}\label{estdino}
\tilde Q(c,g) =\arg\inf_{Q'} V_{c,g}(Q').
\end{equation}
In cases when parameters $c$ or $g$ are unknown, we may plug in
their estimates and define
\begin{eqnarray}\label{estdinopara}
\tilde Q_{\hat c, \hat g}&=&\arg\inf_{Q'} V_{\hat c(Q'), \hat g(Q')}(Q'). 
\end{eqnarray}

\subsubsection{Estimators for the slipping and guessing parameters}\label{SecPara}

To complete our estimation procedure, we provide one generic estimator for $(c,g)$. For the DINA model, we let
\begin{equation}\label{cgdina}
(\hat c(Q), \hat g(Q)) = \arg\inf_{c,g\in [0,1]^m } S_{c,g}(Q);
\end{equation}
and for the DNIO model, we let
\begin{equation}\label{cgdino}
(\hat c(Q), \hat g(Q)) = \arg\inf_{c,g\in [0,1]^m } V_{c,g}(Q).
\end{equation}
We emphasize that $(\hat c, \hat g)$ may not be a consistent
estimator of $(c,g)$. To illustrate this, we present one example
discussed in \cite{LXY2011}. Consider the case of $m=k$ items with $k$ attributes and a complete matrix $Q=\mathcal I_k$, the $k\times k$ identity matrix. The degrees of freedom of a $k$-way binary table is
$2^k-1$. On the other hand, the dimension of parameters $(\mathbf
p, c,g)$ is $2^k-1 +2k$. Therefore, $\mathbf p$, $c$, and $g$ cannot be
consistently identified without additional information. This
problem is typically tackled by introducing addition parametric
assumptions such as $\mathbf p$ satisfying certain functional form
or in the Bayesian setting (weakly) informative prior
distributions  \cite*{Gelman08}. Given that the emphasis of this
paper is the inference of $Q$-matrix, we do not further
investigate the identifiability of $(\mathbf p, c,g)$. Despite the consistency issues, if one adopts the estimators in \eqref{cgdina} and \eqref{cgdino} for the estimator of $Q$ as in \eqref{estc} and \eqref{estdinopara}, the consistency results remain even if $(\hat c(Q), \hat g (Q))$ is inconsistent. We will address this issue in more details in the remarks after the statements of the main theorems.

\subsection{Theoretical properties}

\subsubsection{Notation}\label{Sectheorynotation}

To facilitate the statements, we first introduce notation and some
necessary conditions that will be referred to in later
discussions.
\begin{itemize}
    \item Linear space spanned by vectors $V_1,...,V_l$: $$\mathcal L (V_1,...,V_l)=\left\{\sum_{j=1}^l a_j V_j: a_j \in \mathbb R\right \}.$$

    \item For a matrix $M$, $M_{1:l}$ denotes the submatrix containing the first $l$ rows and all columns of $M$.
    \item Vector $e_i$ denotes a column vector with the $i$-th element being 1 and the rest being 0. When there is no ambiguity, we omit the length index of $e_i$.
    \item Matrix $\mathcal I_l$ denotes the $l\times l$ identity matrix.
    \item For a matrix $M$,  $C(M)$ is the linear space generated by its column vectors. It is usually called the \emph{column space} of $M$.
    \item For a matrix $M$, $C_M$ denotes the set of its column vectors and $R_M$ denotes the set of its row vectors.
    \item Vector $\mathbf 0$ denotes the zero vector, $(0,...,0)$. When there is no ambiguity, we omit the index of length.
    \item Define a $2^k$ dimensional vector
    $$\mathbf p=\left(p_{\mathbf A}: \mathbf A \in \{0,1\}^k\right).$$

    \item For $m$ dimensional vectors $c$ and $g$, write $c\succ g$ if $c_i> g_i$
    for all $1\leq i\leq m$ and $c\ncong g$ if $c_i \neq g_i$ for all $i=1,...,m$.
    \item Matrix $Q$ denotes the true matrix and $Q'$ denotes a generic $m\times k$ binary matrix.
\end{itemize}

The following definitions will be used in subsequent discussions.

\begin{definition}\label{DefSat}
We say that $T(Q)$ is \emph{saturated} if all combinations of the form
$I_{i_1}\wedge ... \wedge I_{i_l}$, for $l=1,...,m$, are included in
$T(Q)$. Similarly, we say that $U(Q)$ is \emph{saturated} if all combinations of the form $I_{i_1}\vee ... \vee I_{i_l}$, for $l=1,...,m$, are included in
$U(Q)$.
\end{definition}

\begin{definition}\label{DefEq}
We write $Q\sim Q'$ if and only if $Q$ and $Q'$ have identical
column vectors, which could be arranged in different orders;
otherwise, we write $Q\nsim Q'$.
\end{definition}

\begin{remark}
It is not hard to show that ``$\sim$'' is an \emph{equivalence
relation}. $Q\sim Q'$ if and only if they are identical after an
appropriate permutation of the columns. Each column of $Q$ is
interpreted as an attribute. Permuting the columns of $Q$ is
equivalent to relabeling the attributes. For $Q\sim Q'$, we are not
able to distinguish $Q$ from $Q'$ based on data.
\end{remark}

\begin{definition}\label{DefComp}
A $Q$-matrix is said to be \emph{complete} if
$\{e_i:i=1,...,k\}\subset R_Q$ ($R_Q$ is the set of row vectors of
$Q$); otherwise, we say that $Q$ is \emph{incomplete}.
\end{definition}

A $Q$-matrix is complete if and only if for each attribute there
exists an item only requiring that attribute. Completeness implies
that $m\geq k$. We will show that completeness is among the
sufficient conditions to identify $Q$. In addition, it is pointed
out by \cite{Chiu} (c.f. the paper for more detailed formulation
and discussion) that the completeness of the $Q$-matrix is a
necessary condition for a set of items to consistently identify
attributes. Thus, it is always recommended to use a complete
$Q$-matrix unless additional information is available.

Listed below are assumptions which will be used in subsequent
development.

\begin{itemize}
\item [C1] Matrix $Q$ is \emph{complete}.
\item [C2] Both $T(Q)$ and $U(Q)$ are \emph{saturated}.
\item [C3] Random vectors $\mathbf A_1,...,\mathbf A_N$ are i.i.d. with the following distribution
    $$P(\mathbf A_r=\mathbf A) = p_{\mathbf A};$$
    We further let $\pp = (p_{\mathbf A}: \mathbf A \in \{0,1\}^k)$.
\item [C4] The attribute population is \emph{diversified}, that is, $\pp \succ \mathbf 0$.
\end{itemize}

\subsubsection{Consistency results}\label{SecConsis}

We first present the consistency results for the DINA model.
\begin{theorem}\label{thmcg}
Under the DINA model, suppose that conditions C1-4 hold, that is,
$Q$ is complete, $T(Q)$ is saturated, the attribute the profiles
are i.i.d., and $\pp$ is diversified.  Suppose also that the $c$
and $g$ are known. Let $S_{c,g}(Q')$ be as defined in
\eqref{scorecg} and
$$\hat Q(c,g) = \arg\inf_{Q'}S_{c,g}(Q').$$
Then,
$$\lim_{N\rightarrow \infty}P(\hat Q(c,g)\sim Q)=1.$$
In addition, with an appropriate arrangement of the column order of $\hat Q(c,g)$, let
$$\hat {\pp} = \arg\inf_{\pp'} |T_{c,g}(\hat Q(c,g))\pp' - \alpha|.$$
Then, for any $\varepsilon>0$,
$$\lim_{N\rightarrow \infty}P(|\hat{\pp}-\pp|>\varepsilon)=0.$$
\end{theorem}

\begin{theorem}\label{thmest}
Under the DINA model, suppose that the conditions in Theorem
\ref{thmcg} hold, except that the $c$ and $g$ are unknown. For any $Q'$, $\hat c(Q')$ and $\hat g (Q')$ are estimators for $c$ and $g$. When $Q= Q'$, $(\hat c(Q), \hat g(Q))$ is a consistent estimator of $(c,g)$. Let
$\hat Q_{\hat c,\hat g}$ be as defined in \eqref{estc}. Then
$$\lim_{N\rightarrow \infty}P(\hat Q_{\hat c,\hat g}\sim Q)=1.$$
In addition, with an appropriate arrangement of the column order of $\hat Q_{\hat c, \hat g}$, let
$$\hat {\pp} = \arg\inf_{\pp'} |T_{\hat c(\hat Q_{\hat c, \hat g}),\hat g(\hat Q_{\hat c, \hat g})}(\hat Q_{\hat c,\hat g})\pp' - \alpha|.$$
Then, for any $\varepsilon>0$,
$$\lim_{N\rightarrow \infty}P(|\hat{\pp}-\pp|>\varepsilon)=0.$$
\end{theorem}

In what follows, we present the consistency results for the DINO model.

\begin{theorem}\label{thmdinocg}
Under the DINO model, suppose that conditions C1-4 hold, that is,
$Q$ is complete,  $U(Q)$ is saturated, the attribute profiles are
i.i.d., and   $\pp$ is diversifies. Suppose also that the $c$
and $g$ are known. Let $V_{c,g}(Q')$ be defined as in
\eqref{ScoreDINO} and
$$\tilde Q(c,g) = \arg\inf_{Q'}V_{c,g}(Q').$$
Then,
$$\lim_{N\rightarrow \infty}P(\tilde Q(c,g)\sim Q)=1.$$
In addition, with an appropriate arrangement of the column order of $\tilde Q(c,g)$, let
$$\hat {\pp} = \arg\inf_{\pp'} |U_{c,g}(\tilde Q(c,g))\pp' - \beta|.$$
Then, for any $\varepsilon>0$,
$$\lim_{N\rightarrow \infty}P(|\hat{\pp}-\pp|>\varepsilon)=0.$$
\end{theorem}

\begin{theorem}\label{thmdinoest}
Under the DINO model, suppose that the conditions in Theorem
\ref{thmdinocg} hold, except that the $c$ and $g$ are unknown.
For any $Q'$, $\hat c(Q')$ and $\hat g (Q')$ are estimators for $c$ and $g$. When $Q= Q'$, $\hat c(Q)$ and $\hat g(Q)$ are  consistent estimators of $c$
and $g$. Let $\hat Q_{\hat c,\hat g}$ be defined as in
\eqref{estdinopara}. Then
$$\lim_{N\rightarrow \infty}P(\tilde Q_{\hat c ,\hat g}\sim Q)=1.$$
In addition, with an appropriate arrangement of the column order of $\tilde Q_{\hat c,\hat g}$, let
$$\hat {\pp} = \arg\inf_{\pp'} |U_{\hat c(\tilde Q_{\hat c,\hat g} ),\hat g(\tilde Q_{\hat c,\hat g} )}(\tilde Q_{\hat c, \hat g})\pp' - \beta|.$$
Then, for any $\varepsilon>0$,
$$\lim_{N\rightarrow \infty}P(|\hat{\pp}-\pp|>\varepsilon)=0.$$
\end{theorem}

\begin{remark}
It is not hard to verify that ``$\sim$'' defines a binary
equivalence relation on the space of $m\times k$ binary matrices,
denoted by $\mathcal M_{m,k}$. As previously mentioned, the data
do not contain information about the specific meaning of the
attributes. Therefore, we do not expect to distinguish $Q_{1}$
from $Q_{2}$ if $Q_{1}\sim Q_{2}$.   Therefore, the
identifiability in the theorems is the strongest type that one may
expect. The corresponding quotient set is the finest resolution
that is possibly identifiable based on the data. Under weaker
conditions, such as in absence of completeness of the $Q$-matrix
or the complete diversity of the attribute distribution, the
identifiability of the $Q$-matrix may be weaker, which corresponds
to a coarser quotient set.
\end{remark}

\begin{remark}
We would like to point out that, when the estimators in
\eqref{cgdina} and \eqref{cgdino} are chosen, $\hat Q_{\hat c,
\hat g}$ is always a consistent estimator of $Q$, even if $(\hat
c, \hat g)$ is not a consistent estimator for $(c,g)$. This is
because the proof of Theorem \ref{thmest} is based on the fact
that $S_{\hat c(Q),\hat g(Q)}(Q)\rightarrow 0$ in probability;
when $Q'\nsim Q$, $S_{\hat c(Q'),\hat g(Q')}(Q')$ is bounded below
by some $\delta
>0$. Given that $S_{c,g} (Q)\rightarrow 0$ and that $(\hat c, \hat
g)$ is chosen to minimize the objective function $S$ , $S_{\hat
c(Q),\hat g(Q)}(Q)$ decreases to zero regardless whether or not
$(\hat c, \hat g)$ is consistent. In addition,  the fact that
$S_{\hat c(Q'),\hat g(Q')}(Q')$ is bounded below by some $\delta
>0$ does not require any consistency property of $(\hat c, \hat
g)$. Therefore, the consistency of $\hat Q_{\hat c, \hat g}$ does
not rely on the consistency of $(\hat c, \hat g)$ if it is of the
particular forms as in \eqref{cgdina} and \eqref{cgdino}. On the
other hand, in order to have $\hat {\pp}$ being consistent, it is
necessary to require the consistency for $(\hat c, \hat g)$.
Therefore, in the statement of Theorem \ref{thmest} we require the
consistency of $(\hat c, \hat g)$, though it is necessary to point
out this subtlety. A similar argument applies to Theorem
\ref{thmdinoest} as well.
\end{remark}



\section{Discussions and implementation}\label{SecDisc}

This paper focuses mostly on the estimation  of the $Q$-matrix. In
this section, we discuss several practical issues and a few other
usages of the proposed tools.

\paragraph{Computational issues.}
There are several aspects we would like to address. First, for a
given $Q$, the evaluation of $S_{c,g}(Q)$ only consists of
optimization of a quadratic function subject to linear
constraint(s). This can be done by quadratic programming type of
well established algorithms.

Second, the theories require construction of a saturated
$T$-matrix or $U$-matrix which is $2^m-1$ by $2^k$. Note that when
$m$ is reasonably large, for instance, $m=20$, a saturated
$T$-matrix has over 1 million rows. One solution is to include
part of the combinations and gradually include more combinations
if the criterion function admit small values at multiple
$Q$-matrices. Alternatively, we may split the items into multiple
groups which we will elaborate in the next paragraph.

The third computational issue is related to minimization of
$S_{c,g}(Q)$ with respect to $Q$. This involves evaluating
function $S$ over all the $m\times k$ binary matrices, which has a
cardinality of $2^{m\times k}$. Simply searching through such a
space is a substantial computation overhead. In practice, one may
want to handle such a situation by splitting the $Q$-matrix in the
following manner. Suppose there are $m$ items. We split them into
$l$ groups, each of which has $m_0$ (a computationally manageable
number) items.  This is equivalent to dividing a large $Q$-matrix
into multiple smaller sub-matrices. When necessary, we may allow
different groups to have overlaps of items. Then, we can estimate
each sub-matrix separately and merge them into an estimate of the
big $Q$-matrix. Given that the asymptotic results are applicable
to each of the sub-matrices, the combined estimate is also
consistent. This is similar to the splitting procedure in Chapter
8.6 of \cite{Tatsuoka}. We emphasize that splitting the parameter
space is typically not valid for usual statistical inferences.
However, the $Q$-matrix admits a special structure with which the
splitting is feasible and valid. This partially helps to relieve
the computation burden related to the proposed procedure. On the
other hand, it is always desirable to have a generic efficient
algorithm for a general large scale $Q$-matrix. We leave this as a
topic for a future investigation.

\paragraph{Partially specified $Q$-matrix.} It is often reasonable
to assume that some entries of the $Q$-matrix are known. For
example, suppose we can separate the attributes into ``hard'' and
``soft'' ones. By ``hard'', we mean those that are concrete and
easily recognizable in a given problem and, by ``soft'', we mean
those that are subtle and not obvious. We can then assume that the
entry columns which correspond to the ``hard'' attributes are
known. Another instance is that there is a subset of items whose attribute requirements are known and the item-attribute relationships of the other items need to be learnt, such as the scenarios when new items need to be calibrated according to the existing ones.
In this sense, even if an estimated $Q$-matrix may not be sufficient to replace the a priori $Q$-matrix provided by the ``expert'' (such as exam makers), it can serve as a validation as well as a source of calibration of the existing knowledge of the $Q$-matrix.

When such information is available and correct, the computation
can be substantially reduced. This is because the optimization,
for instance that in \eqref{estc}, can be performed subject to the
existing knowledge of the $Q$-matrix. In particular, once a set of
items is known to form a complete $Q$-matrix, that is, item $i$ is
known to only require attribute $i$ for $i=1,...,k$, then one can
calibrate one item at a time. More specifically, at each time, one
can estimate the sub-matrix consisting of items $1$ to $k$ as well
as one additional item, the computational cost of which is
$O(2^k)$. Then the overall computational cost is reduced to $O(m
2^k)$, which is typically of a manageable order.

\paragraph{Validation of a $Q$-matrix.}
The propose framework is applicable to not only the estimation of
the $Q$-matrix but also validation of an existing $Q$-matrix.
Consider the DINA and DINO models. If the $Q$-matrix is correctly
specified, then one may expect
$$|\alpha - T_{\hat c, \hat g}(Q) \pp| \rightarrow 0$$
in probability as $N\rightarrow \infty$. The above convergence
requires no additional conditions (such as completeness or
diversified attribute distribution). In fact, it suffices to have
that the responses are conditionally independent given the
attributes and $(\hat c, \hat g)$ are consistent estimators of
$(c,g)$. Then, one may expect that
$$S_{\hat c, \hat g}(Q) \rightarrow 0.$$
If the convergence rate of the estimators $(\hat c,\hat g)$ is
known, for instance, $(\hat c -c, \hat g- g)= O_{p}(n^{-1/2})$,
then a necessary condition for a correctly specified $Q$-matrix is
that $S_{\hat c, \hat g}(Q)= O_{p}(n^{-1/2})$. The asymptotic
distribution of $S$ depends on the specific form of $(\hat c,
\hat g)$. Consequently, checking the closeness of $S$ to zero
forms a procedure for validation of the existing knowledge of the
$Q$-matrix.

\section{Proofs of the theorems}\label{SecProof}

\subsection{Preliminary results: propositions and lemmas}\label{SecLem}

\begin{proposition}\label{PropRank}
Under the setting of the DINA model, suppose that $Q$ is complete and matrix $T(Q)$ is saturated. Then,
we are able to arrange the columns and rows of $Q$ and $T(Q)$ such
that  $T(Q)_{1:(2^k-1)}$ has rank $2^k-1$, that is, after removing
one zero column this sub-matrix has full column rank.
\end{proposition}

\begin{proof}[Proof of Proposition \ref{PropRank}]
We let the first column of $T(Q)$ correspond to the zero attribute
profile. Then, the first column is a zero vector, which is the
column we mean to remove in the statement of the proposition.
Provided that $Q$ is complete, without loss of generality we
assume that the $i$-th row vector of $Q$ is $e_i^\top$ for
$i=1,...,k$, that is, item $i$ only requires attribute $i$ for
each $i=1,...,k$. The first $2^k -1$ rows of $T(Q)$ are associated
with $\{I_1,...,I_k\}$. In particular, we let the first $k$ rows
correspond to $I_1,...,I_k$ and the second to the $(k+1)$-th
columns of $T(Q)$ correspond to $\mathbf A$'s that only have one
attribute. We
further arrange the next $C^k_2$ 
rows of $T(Q)$ to correspond to combinations of two items, $I_i
\wedge I_j$, $i\neq j$. The next $C^k_2$ columns of $T(Q)$
correspond to $\mathbf A$'s that only have two positive attributes.
Similarly, we arrange $T(Q)$ for combinations of three, four, and up
to $k$ items. Therefore, the first $2^{k}-1$ rows of $T(Q)$ admit a
block upper triangle form. In addition, we are able to further
arrange the columns within each block such that the diagonal
matrices are identities, so that $T(Q)$ has form
\begin{equation}\label{A}
\begin{array}{c}
I_{1},I_{2},... \\
I_{1}\wedge I_{2},I_{1}\wedge I_{3},... \\
I_{1}\wedge I_{2}\wedge I_{3},... \\
\vdots%
\end{array}%
\left(
\begin{array}{ccccccccccc}
0&\mathcal I_{k} & \ast  & \ast  & \ast  & \ldots &  \\
0&0 & \mathcal I_{C_{2}^{k}} & \ast  & \ast  &  &    \\
0&0 & 0 & \mathcal I_{C_{3}^{k}} & \ast  &  &    \\
\vdots&\vdots & \vdots & \vdots &  &  &
\end{array}%
\right).
\end{equation}
$T(Q)_{1:(2^k-1)}$ obviously has full rank after removing the zero (first) column.
\end{proof}

From now on, we assume that $Q_{1:k}=\mathcal I_k$ and the first
$2^{k}-1$ rows of $T(Q)$ are arranged in the order as in \eqref{A}.

\begin{proposition}\label{PropComp}
Under the DINA model, that is, the ability indicator follows
\eqref{xidina}, assume that $Q$ is a complete matrix and $T(Q)$ is
saturated. Without loss of generality, let $Q_{1:k} =\mathcal
I_k$. Assume that the first $k$ rows of $Q'$ form a complete
matrix. Further, assume that $Q_{1:k}=Q'_{1:k}=\mathcal I_k$. If
$Q'\neq Q$ and $c\ncong g$, then for all $c'\in \mathbb R^m$ there
exists at least one column vector of $T_{c,g}(Q)$ not in the
column space $C(T_{c'}(Q'))$, where $T_{c'}(Q')$ is as defined in
\eqref{Tc} being the $T$-matrix with zero guessing probabilities.
\end{proposition}

\begin{proposition}\label{PropIncomp}
Under the DINA model, that is, the ability indicator follows
\eqref{xidina}, assume that $Q$ is a complete matrix and $T(Q)$ is
saturated. Without loss of generality, let $Q_{1:k} =\mathcal
I_k$. If $c\ncong g$ and $Q'_{1:k}$ is incomplete, then for all
$c' \in \mathbb R^m$ there exists at least one nonzero column vector of
$T_{c,g}(Q)$ not in the column space $C(T_{c'}(Q'))$.
\end{proposition}

In the statement of Propositions \ref{PropComp} and \ref{PropIncomp}, $c_i$, $g_i$, and $c'_i$ can be any real numbers and are not restricted to be in $[0,1]$.
Propositions \ref{PropComp} and \ref{PropIncomp} are the central
results of this paper, whose proofs are delayed to the Appendix.
To state the next proposition, we define matrix
\begin{equation}\label{tildeT}
\tilde{T}_{c,g}(Q)=\left(
\begin{array}{c}
T_{c,g}(Q) \\
\EE
\end{array}
\right) ,
\end{equation}
that is, we add one more row of one's to the original $T$-matrix.

\begin{proposition}
\label{PropSG}Under the DINA model, that is, the ability indicator
follows \eqref{xidina}, suppose that $Q$ is a complete matrix,
$Q'\nsim Q$, $T$ is saturated, and $c\ncong g$. Then,  for all
$c,g, c', g'\in [0,1]^m$, there exists one column vector of
$\tilde{T}_{c,g}(Q)$ (depending on $c,g,c',g'$) not in $C(\tilde
T_{c',g'}(Q'))$. In addition, $\tilde T_{c,g}(Q)$ is of full
column rank.
\end{proposition}

\begin{lemma}
\label{LemColT}Consider two matrices $T_{1}$ and $T_{2}$ of the same
dimension. If $C(T_{1})\subseteq C(T_{2})$, then for any matrix $D$
of appropriate dimension for multiplication, we have%
\begin{equation*}
C(DT_{1})\subseteq C(DT_{2}).
\end{equation*}

Conversely, if the $l$-th column vector of $DT_{1}$ does not belong
to $C(DT_{2})$, then the $l$-th column vector of $T_{1}$ does not belong to $%
C(T_{2})$.
\end{lemma}

\begin{proof}[Proof of Lemma \protect\ref{LemColT}]
Note that $DT_i$ is just a linear row transform of $T_i$ for
$i=1,2$. The conclusion is immediate by basic linear algebra.
\end{proof}

\bigskip

\begin{proof}[Proof of Proposition \ref{PropSG}]
According to Propositions \ref{PropComp} and \ref{PropIncomp} and
Lemma \ref{LemColT}, it is sufficient to show that there exists a
matrix $D$ such that
$$D \tilde T_{c,g}(Q) = T_{c-g',g-g'}(Q), \quad D \tilde T_{c',g'}(Q') = T_{c'-g',\mathbf 0} (Q')\triangleq  T_{c'_{g'}}(Q'),$$
where $c'_{g'}= c'-g'$.
Once we obtain such a linear transformation, according to Propositions \ref{PropComp} and \ref{PropIncomp}, there exists a column vector in $T_{c-g',g-g'}(Q)$ that is not in the column space of $T_{c'_{g'}}(Q')$, as long as $Q\nsim Q'$. Then the same column vector in $\tilde T_{c,g}(Q)$ is not in the column space of $\tilde T_{c',g'}(Q')$. Thereby, we are able to conclude the proof.

In what follows, we construct such a $D$ matrix. Let $g^*=
(g_1^*,...,g_m^*)$. We show that there exists a matrix $D_{g^*}$
only depending on $g^*$ so that $D_{g^*}\tilde T_{c,g}(Q) =
T_{c-g^*,g-g^*}(Q)$. Note that each row of $D_{g^*}\tilde
T_{c,g}(Q)$ is just a row linear transform of $ \tilde
T_{c,g}(Q)$. Then, it is sufficient to show that each row vector
of $ T_{c-g^*,g-g^*}(Q)$ is a linear transform of rows of
$\tilde T_{c,g}(Q)$ with coefficients only depending on
$g^*$. We prove this by induction.

First, note that
$$B_{c-g^*, g-g^*,Q}(I_i) = B_{c,g,Q}(I_i) -  g_i^* \EE.$$
Then all row vectors of $T_{c-g^*,g-g^*}(Q)$ of the form
$B_{c-g^*, g-g^*,Q}(I_i) $ are inside the row space of $\tilde
T_{c,g}(Q)$ with coefficients only depending on $g^*$. Suppose
that all the vectors of the form
$$B_{c-g^*, g-g^*, Q}(I_{i_1}\wedge ...\wedge I_{i_l})$$
for all $1\leq l\leq j$ can be written linear combinations of the
row vectors of $ \tilde T_{c,g}(Q)$ with coefficients only depending on
$g^*$. Then, we consider
$$B_{c, g, Q}(I_{i_1}\wedge ...\wedge I_{i_{j+1}}) = \Upsilon_{h=1}^{j+1}\left(B_{c-g^*, g-g^*, Q}(I_{i_h}) + g_{i_h}^* \EE\right ).$$
The left hand side is just a row vector of $\tilde T_{c,g}(Q)$. We expand the right hand side of the above display. Note that the last term is precisely
$$B_{c-g^*, g-g^*, Q}(I_{i_1}\wedge ...\wedge I_{i_{j+1}})=\Upsilon_{h=1}^{j+1} B_{c-g^*, g-g^*, Q}(I_{i_h}) .$$
The rest terms are all of the form $B_{c-g^*, g-g^*, Q}(I_{i_1}\wedge ...\wedge I_{i_l})$ for $1\leq l\leq j$ multiplied by  coefficients only depending on $g^{*}$. Therefore, according to the induction assumption, we have that
$$B_{c-g^*, g-g^*, Q}(I_{i_1}\wedge ...\wedge I_{i_{j+1}})$$
can be written as linear combinations of rows of $\tilde
T_{c,g}(Q)$ with coefficients only depending on $g^*$. Therefore,
we can construct the matrix $D_{g^*}$ accordingly. Lastly, we choose
$g^* = g'$ and conclude that
$$D_{g'} \tilde T_{c,g}(Q) = T_{c-g',g-g'}(Q), \quad D_{g'} \tilde T_{c',g'}(Q') = T_{c'_{g'}}(Q').$$
By Propositions \ref{PropComp} and \ref{PropIncomp}, there exists a column vector of $T_{c-g',g-g'}(Q)$ not in the column space of $T_{c'_{g'}}(Q')$. Furthermore, according to Lemma \ref{LemColT}, we conclude the first part of the Proposition.

In addition, consider $D_{g}\tilde T_{c,g}(Q) = T_{c_{g}} (Q)$ where $c_{g} = c-g \ncong \mathbf 0$. By construction as in \eqref{A}, after removing the first zero column, $T_{c_{g}} (Q)$ is of rank $2^{k}-1$. Therefore, the matrix
\begin{equation*}
\left(\begin{array}{c}
T_{c_{g}} (Q)\\
\mathbf E
\end{array}\right)
\end{equation*}
is of full rank. Note that each row of the above matrix is a linear transform of $\tilde T_{c,g}(Q)$. Thus, $\tilde T_{c,g}(Q)$ is a full rank matrix too.
Thereby, we conclude the proof of the proposition.
\end{proof}

%

For the DINO model, we define a similar matrix
\begin{equation}\label{tildeTstar}
\tilde{U}_{c,g}(Q)=\left(
\begin{array}{c}
U_{c,g}(Q) \\
\EE
\end{array}
\right) ,
\end{equation}
and collect the following proposition.

\begin{proposition}
\label{PropDINO}Under the setting of the DINO model, that is, the
ability indicator follows \eqref{xidino}, suppose that $Q$ is a
complete matrix, $Q'\nsim Q$, $U$ is saturated, and $c\ncong g$.
Then, for all $c,g,c',g'\in [0,1]^m$, there exists one column
vector of $\tilde{U}_{c,g}(Q)$ not in $C(\tilde U_{c',g'}(Q'))$. In
addition, $\tilde U_{c,g}(Q)$ is of full column rank.
\end{proposition}

\begin{lemma}\label{PropDual} Let $T(Q)$ be the $T$-matrix under the
DINA model with $c=1$ and $g=0$ and $U(Q)$ be the $U$-matrix
under DINO model  with $c=1$ and $g=0$. We are able to arrange the
column order of $T(Q)$ and $U(Q)$ so that
$$T(Q) + U(Q) = \mathbb E,$$
where $\mathbb E$ is a matrix of appropriate dimensions with all entries being one's.
\end{lemma}

\begin{proof}[Proof of Lemma \ref{PropDual}]
Consider a  $Q$-matrix, an attribute profile $\AA$, and an item
$i$. Let $\AA^c = \EE -\AA$ be the complimentary profile. Suppose
that $Q_{ij}=1$ for $1\leq j\leq n$ and $Q_{ij}=0$ for $n< j \leq
k$. Under the DINO model, $\xi_{DINO}^i(\AA,Q)=1$ if $\AA^j =1$ at
least for one $1\leq j \leq n$.  For the same $j$, $(\AA^c)^j =0$
and therefore $\xi_{DINA}^i(\AA^c, Q)=0$. That is,
$\xi^{i}_{DINO}(\AA,Q) = 1$ implies that
$\xi^i_{DINA}(\AA^{c},Q)=0$. Similarly we are able to obtain that
$\xi^{i}_{DINO}(\AA,Q) = 0$ implies that
$\xi^i_{DINA}(\AA^{c},Q)=1$. Therefore, if we arrange the columns
of $T(Q)$ and $U(Q)$ in such a way that the $\AA$-th column of
$U(Q)$ and the $\AA^c$-th column of $T(Q)$ have the same position,
then
$$B_Q(I_i) + F_Q (I_i) = \EE, $$
for all $1\leq i \leq m$. Note that
\begin{eqnarray*}B_Q (I_1\wedge ... \wedge I_l ) &=&\Upsilon_{i=1}^l B_Q (I_i)\\
&=&\Upsilon_{i=1}^l (\EE - F_Q (I_i))\\
&=&\EE - F_Q(I_1\vee...\vee I_l).
\end{eqnarray*}
Thus, we conclude the proof.
\end{proof}

\bigskip

\begin{proof}[Proof of Proposition \ref{PropDINO}]
Thanks to Propositions \ref{PropComp} and \ref{PropIncomp} and
Lemma \ref{LemColT}, it is sufficient to show that with an
appropriate order of the columns of $U_{c,g}(Q)$ there exists a
matrix $D'_{c'}$ only depending on $c'=(c'_1,...,c'_m)$ (independent of $Q$) such that
$$D'_{c'} \tilde U_{c,g}(Q) = T_{c'-g,c'-c}(Q) $$
for all $m\times k$ binary matrix $Q$. To establish that, we only need to show
that each row vector of $T_{c'-g,c'-c} (Q)$ can be written as a linear
combination of the row vectors of $\tilde U_{c,g}(Q)$. In addition, the
coefficients only depend on the $c'$ and are independent of
$c$, $g$, and $Q$.

We establish this by induction. By construction, we have that for each $i=1,...,m$
$$\EE - F_{c,g,Q}(I_i) = (1-c_i)\EE + (c_i - g_i) (\EE - F_Q (I_i)). $$
Note that each column of $U$ (and $T$) and each element in $F_Q(I_i)$
(and $B_Q(I_i)$) correspond to one attribute profile $\AA\in
\{0,1\}^k$. If we arrange the $\AA$-th position of $F_Q(I_i)$ and
$\AA^c$ position of $B_Q (I_i)$ to be the same, then from the
proof of Lemma \ref{PropDual} we obtain that $B_Q(I_i) = \EE - F_Q
(I_i)$. Therefore, $\EE - F_{c,g,Q}(I_i)  = B_{1-g,1-c,Q} (I_i)$.
Similarly, we obtain that
\begin{eqnarray*}\EE - F_{c,g,Q}(I_{i_1}\vee...\vee I_{i_{l}}) &=& \Upsilon_{j=1}^{l} (\EE - F_{c,g,Q}(I_{i_j}))\\
&=&\Upsilon_{j=1}^{l+1}B_{1-g,1-c,Q}(I_{i_j})\\
&=&B_{1-g,1-c,Q}(I_{i_1}\wedge...\wedge I_{i_{l}}),
\end{eqnarray*}
where $1-c = (1-c_1,...,1-c_m)$.
Let $\mathbb E$ be the matrix with all entries being one's. We essentially established that
$$\mathbb E - U_{c,g}(Q) = T_{1-g, 1-c}(Q).$$
We use the matrix $D_{g^*}$ constructed in Proposition \ref{PropSG} and obtain that
$$D_{1 -c'} \left(\begin{array}{c}\mathbb E - U_{c,g}(Q)\\
\mathbf E
\end{array}\right) = D_{1-c'} \tilde T_{1-g, 1-c}(Q)= T_{c'-g,c'-c }(Q).$$
Similarly, we have that
$$D_{1 -c'}
\left(\begin{array}{c}\mathbb E - U_{c',g'}(Q')\\
\mathbf E
\end{array}\right)= D_{1 -c'} \tilde T_{1-g',1-c'}(Q')
= T_{c'-g' }(Q').$$
Note that $\EE$ is a row vector of both $\tilde U_{c,g}(Q)$ and $\tilde U_{c',g'}(Q')$. Therefore, one can construct a matrix $D'_{c'}$ so that
$$D'_{c'} \tilde U_{c,g}(Q)=  T_{c'-g,c'-c }(Q), \quad  D'_{c'} \tilde U_{c',g'}(Q')=  T_{c'-g' }(Q').$$

Thanks to Propositions \ref{PropComp} and \ref{PropIncomp}, there
exists a column vector of $T_{c'-g,c'-c}(Q)$ not inside the column
space of $T_{c'-g'}(Q')$ whenever $c\ncong g$. Thanks to Lemmas
\ref{LemColT} and \ref{PropDual}, the corresponding column vector(s) of $\tilde
U_{c,g}(Q)$ is not inside the column space of $\tilde
U_{c',g'}(Q')$. In addition, note that
\begin{equation*}
\left(
\begin{array}{c}
T_{c'-g,c'-c}(Q) \\
\EE
\end{array}
\right)
\end{equation*}
is of full column rank (Proposition \ref{PropSG}) and can be obtained by a row transformation
of $\tilde U_{c,g}(Q)$. Therefore, $\tilde U_{c,g}(Q)$ is also of
full column rank. Thereby, we conclude the proof.
\end{proof}

\subsection{Proof of the theorems}

\begin{proof}[Proof of Theorem \ref{thmcg}]
Notice that the true parameters $c$ and $g$  form  consistent
estimators for themselves. Therefore, Theorem \ref{thmcg} is a
direct corollary of Theorem \ref{thmest}.
\end{proof}

\bigskip

\begin{proof}[Proof of Theorem \ref{thmest}]
By the law of large numbers,
\begin{equation*}
|T_{c,g}(Q)\mathbf p-\alpha|\rightarrow 0
\end{equation*}%
almost surely as $N\rightarrow \infty $. Therefore,
\begin{equation*}
S_{c,g}(Q)\rightarrow 0
\end{equation*}%
almost surely as $N\rightarrow \infty $. Note that $S_{c,g}(Q)$ is
a continuous function of $(c,g)$. The consistency of $(\hat c(Q),
\hat g(Q))$ implies that
$$S_{\hat c(Q), \hat g(Q)} (Q) \rightarrow 0,$$
in probability as $N\rightarrow \infty$.

For any $Q^{\prime }\nsim Q$, note that
\begin{equation*}
\left(\begin{array}{c}
\alpha  \\
1%
\end{array}%
\right) \rightarrow \tilde{T}_{c,g}(Q)\pp
\end{equation*}%
According to Proposition \ref%
{PropSG} and the fact that $\pp\succ \mathbf{0}$, there exists $\delta (c',g')>0$ such that $\delta (c',g')$ is continuous in $(c',g')$ and
\begin{equation*}
\inf_{\pp'}\left\vert \tilde{T}_{c',g'}(Q^{\prime })\pp' -\tilde{T}_{c,g}(Q)\pp \right\vert >\delta (c',g').
\end{equation*}%
By elementary calculus,
\begin{equation*}
\delta \triangleq \inf_{c',g'\in [ 0,1]^{m}}\delta
(c',g')>0
\end{equation*}%
and%
\begin{equation*}
\inf_{c',g'\in [0,1]^m,\pp'}\left\vert \tilde{T}_{c',g'}(Q^{\prime })\pp' -\tilde{T}_{c,g}(Q)\pp \right\vert  >\delta .
\end{equation*}%
Therefore,
\begin{equation*}
P\left(\inf_{c',g'\in [0,1]^m,\pp'}\left\vert \tilde{T}%
_{c',g'}(Q^{\prime })\pp'-\left(
\begin{array}{c}
\alpha  \\
1%
\end{array}%
\right) \right\vert >\delta/2 \right) \rightarrow 1,
\end{equation*}%
as $N\rightarrow \infty $. For the same $\delta $, we have%
\begin{equation*}
P(S_{\hat c (Q'),\hat g(Q')}(Q') > \delta/2)\geq P(\inf_{c',g'\in [0,1]^m}S_{c',g'}(Q')> \delta/2)=P\left(\inf_{c',g'\in [0,1]^m,\pp'}\left\vert T
_{c',g'}(Q^{\prime })\pp'-\alpha \right\vert >\delta/2 \right)
 \rightarrow 1.
\end{equation*}%
The above minimization in the last probability is subject to the
constraint that
\begin{equation*}
\sum_{\mathbf{A}\in \{0,1\}^{k}}p_{\mathbf{A}}=1.
\end{equation*}%
Together with the fact that there are only finitely many $m\times k$
binary matrices, we have
\begin{equation*}
P(\hat{Q}_{\hat c, \hat g}\sim Q)=1.
\end{equation*}%
We arrange the columns of $\hat{Q}_{\hat c,\hat g}$ so that $P(\hat{Q}_{\hat c, \hat g}=Q)\rightarrow 1$ as $N\rightarrow \infty $.

Now we proceed to the proof of consistency for
$\hat {\pp}$. Note
that%
\begin{eqnarray*}
\left\vert \tilde{T}_{\hat c(\hat Q_{\hat c, \hat g}),\hat g(\hat Q_{\hat c, \hat g})}(\hat Q_{\hat c, \hat g})\hat {\pp}-\left(
\begin{array}{c}
\alpha \\
1%
\end{array}%
\right) \right\vert &\overset{p}{\rightarrow }&0, \\
\left\vert \tilde{T}_{\hat c(Q),\hat g(Q)}(Q)\pp-\left(
\begin{array}{c}
\alpha \\
1%
\end{array}%
\right) \right\vert &\overset{p}{\rightarrow }&0.
\end{eqnarray*}%
Note that  $\tilde{T}_{c,g}(Q)$ is a full column rank matrix, $P(\hat Q_{\hat c,\hat g}=Q)\rightarrow 1$, $\hat c (Q)\rightarrow c$, $\hat g(Q) \rightarrow g$, and $T_{c,g}$ is continuous in $(c,g)$. Then, we obtain that $\hat {\pp}\rightarrow \mathbf
p$ in probability.
\end{proof}

\bigskip

\begin{proof}[Proof of Theorem \ref{thmdinocg}]
Similar to Theorem \ref{thmcg}, Theorem \ref{thmdinocg} is a direct corollary of Theorem \ref{thmdinoest}.
\end{proof}

\bigskip

\begin{proof}[Proof of Theorem \ref{thmdinoest}]
The proof of Theorem \ref{thmdinoest} is completely analogous to that of Theorem \ref{thmest}. Therefore, we omit the details.
\end{proof}

\appendix

\section{Technical proofs}\label{SecTech}

\begin{proof}[Proof of Proposition \protect\ref{PropComp}]
Note that $Q_{1:k}=Q_{1:k}^{\prime }=\mathcal{I}_{k}$. Let
$T(\cdot )$ be arranged as in \eqref{A}. Then,
$T(Q)_{1:(2^{k}-1)}=T(Q^{\prime })_{1:(2^{k}-1)}$. Given that
$Q\neq Q^{\prime }$, we have $T(Q)\neq T(Q^{\prime })$. We assume
that $T(Q)_{li}\neq T(Q')_{li}$, where $T(Q)_{li}$ is the entry in
the $l$-th row and $i$-th column. Since $
T(Q)_{1:(2^{k}-1)}=T(Q^{\prime })_{1:(2^{k}-1)}$, it is necessary
that $l\geq 2^{k}$. In addition, we let the $l$-th row correspond
to a single item (not combinations of multiples).

Suppose that the $l$-th row of the $T(Q^{\prime })$ corresponds to
an item
that requires attributes $i_{1},...,i_{l^{\prime }}$. Then, we consider $%
1\leq h\leq 2^{k}-1$, such that the $h$-th row of $T(Q^{\prime })$ is $%
B_{Q^{\prime }}(I_{i_{1}}\wedge ...\wedge I_{i_{l^{\prime }}})$.
Then, the $h $-th row vector and the $l$-th row vector of
$T(Q^{\prime })$ are identical.

Since $T(Q)_{1:(2^{k}-1)}=T(Q^{\prime })_{1:(2^{k}-1)}$, we have $%
T(Q)_{hj}=T(Q^{\prime })_{hj}=T(Q^{\prime })_{lj}$ for
$j=1,...,2^{k}-1$. If $T(Q)_{li}=0$ and $T(Q^{\prime })_{li}=1$, the
matrices $T(Q)$ and $T(Q^{\prime })$ look like

\begin{singlespace}
\begin{eqnarray*}
&&%
\begin{array}{ccccccccccccc}
&  &  &  & \text{ \ } &  &  &
\begin{array}{c}
\mbox{column }i \\
\downarrow \quad
\end{array}
&  &  &  &  &
\end{array}
\\
T(Q^{\prime }) &=&%
\begin{array}{c}
\\
\mbox{row }h\rightarrow  \\
\\
\\
\mbox{row }l\rightarrow  \\
\end{array}%
\left(
\begin{array}{cccccc}
0 & \mathcal{I} & \ast  & \ldots  & \ast  & \ldots  \\
\vdots  & \vdots  & \vdots  &  & \ldots  & \ldots  \\
\vdots  & \vdots  & \vdots  & \mathcal{I} & \ldots  & \ldots  \\
\vdots  & \vdots  & \vdots  & \vdots  &  &  \\
0 & \ast  & 1 & \ast  &  &  \\
0 & \ast  & \ast  & \ast  &  &
\end{array}%
\right) ,
\end{eqnarray*}
and
\begin{eqnarray*}
&&%
\begin{array}{ccccccccccccc}
&  &  &  &  &  & \text{ \ } &
\begin{array}{c}
\mbox{column }i \\
\downarrow \quad
\end{array}
&  &  &  &  &
\end{array}
\\
T(Q) &=&%
\begin{array}{c}
\\
\mbox{row }h\rightarrow  \\
\\
\\
\mbox{row }l\rightarrow  \\
\end{array}%
\left(
\begin{array}{cccccc}
0 & \mathcal{I} & \ast  & \ldots  & \ast  & \ldots  \\
\vdots  & \vdots  & \vdots  &  & \ldots  & \ldots  \\
\vdots  & \vdots  & \vdots  & \mathcal{I} & \ldots  & \ldots  \\
\vdots  & \vdots  & \vdots  & \vdots  &  &  \\
0 & \ast  & 0 & \ast  &  &  \\
0 & \ast  & \ast  & \ast  &  &
\end{array}%
\right) .
\end{eqnarray*}
\end{singlespace}

\begin{enumerate}
\item[Case 1] The $h$-th and $l$-th row vectors of $T_{c^{\prime }}(Q^{\prime })$ are
nonzero vectors.

Consider the following two submatrices%
\begin{equation*}
M_{1}=\left(
\begin{array}{cc}
T_{c,g}(Q)_{hi} & T_{c,g}(Q)_{h2^{k}} \\
T_{c,g}(Q)_{li} & T_{c,g}(Q)_{l2^{k}}%
\end{array}%
\right) ,M_{2}=\left(
\begin{array}{ccc}
T_{c^{\prime }}(Q^{\prime })_{h1} & ... & T_{c^{\prime }}(Q^{\prime
})_{h2^{k}} \\
T_{c^{\prime }}(Q^{\prime })_{l1} & ... & T_{c^{\prime }}(Q^{\prime
})_{l2^{k}}%
\end{array}%
\right) .
\end{equation*}%
By construction that $T(Q')_{hi} = T(Q')_{li}$ for all $i$, all column vectors of $M_{2}$ are proportional to each
other. In what follows, we identify one column of $T_{c,g}(Q)$ that is not in the column space of $T_{c'}(Q')$. Also, it is useful to keep in mind that the $2^{k}$-th (last) column of $T$ corresponds to the attribute profile $(1,...,1)$.

\begin{itemize}
\item[a1] If $T(Q)_{li}=0$ and $T(Q)_{hi}=1$, then
$T_{c,g}(Q)_{hi}=T_{c,g}(Q)_{h2^{k}}$. Since $c\ncong g$, we
obtain that $T_{c,g}(Q)_{li}\neq T_{c,g}(Q)_{l2^{k}}.$ There are
two situations:

\begin{itemize}
\item [b1] $T_{c,g}(Q)_{hi}=T_{c,g}(Q)_{h2^{k}}\neq 0$. It is
straightforward to see that the column space of $M_{2}$ does not
contain both column vectors of $M_{1}$. This is because
$T_{c,g}(Q)_{hi}=T_{c,g}(Q)_{h2^{k}}\neq 0$ and
$T_{c,g}(Q)_{li}\neq T_{c,g}(Q)_{l2^{k}}$ imply that the two
column vectors of $M_{1}$ are not proportional to each other. Then,
either the $i$-th column or the $2^{k}$-th column of $T_{c,g}(Q)$
is not in the column space of $T_{c'}(Q')$.

\item [b2] $T_{c,g}(Q)_{hi}=T_{c,g}(Q)_{h2^{k}}=0$.
$T_{c,g}(Q)_{li}\neq T_{c,g}(Q)_{l2^{k}}$ implies that at least
one of them is nonzero. Suppose that $T_{c,g}(Q)_{li} \neq 0$,
then the $i$-th column of $T_{c,g}(Q)$ is not in the column space
of $T_{c'}(Q')$. This is because the $h$-th row of $T_{c^{\prime
}}(Q^{\prime })$ is not a zero vector and any vector of the form
\begin{equation}
\left(
\begin{array}{c}
0\\nonzero
\end{array}
\right)
\end{equation}
is not in the column space of the $M_{2}$. Similarly, if $T_{c,g}(Q)_{l2^{k}} \neq 0$, then the $2^{k}$-th column is identified.

\end{itemize}

\item[a2]
If $T(Q)_{li}=1$ and $T(Q)_{hi}=0$, then $T_{c,g}(Q)_{li}=T_{c,g}(Q)_{l2^{k}}
$. Note that row $h$ corresponds to a combination of items (or just one item)
each of which only requires one attribute. Therefore, we may choose column $i
$ such that the corresponding attribute is capable of answering all items in row $h$
except for one. With this construction, if $T_{c,g}(Q)_{hi}= T_{c,g}(Q)_{h2^k}$, then they must be both zero (most of the time $T_{c,g}(Q)_{hi}$ and  $T_{c,g}(Q)_{h2^k}$ are distinct). We consider three situations:
    \begin{itemize}
    \item [c1]$T_{c,g}(Q)_{hi}\neq T_{c,g}(Q)_{h2^{k}}$. Similar to a1, the conclusion is straightforward.

    \item [c2]$T_{c,g}(Q)_{hi}=T_{c,g}(Q)_{h2^{k}}=0$ and $T_{c,g}(Q)_{li}=T_{c,g}(Q)_{l2^{k}}\neq 0$. Similar to b2, since the $h$-th row vector of $T_{c^{\prime }}(Q^{\prime })$ is nonzero, the statement of the proposition also holds.

    \item [c3]$T_{c,g}(Q)_{hi}=T_{c,g}(Q)_{h2^{k}}=0$ and $
    T_{c,g}(Q)_{li}=T_{c,g}(Q)_{l2^{k}}=0$. This situation is slightly complicated, since $M_{1}$ is a zero matrix and we have to seek for a different column other than the columns $i$ and $2^k$. In what follows, all the item-attribute relationship refers to $Q$. If the item in the $l$-th row does not require strictly fewer attributes than the items in row $h$, then, we are able to find a column as in a1.

    \bigskip

    Otherwise, the item in the $l$-th DOES require strictly fewer attributes than the items in row $h$. Without loss of generality, assume that the item corresponding to the $l$-th row requires attribute $1, 2,..., j$, and the $h$-th row corresponds to items $1, 2,...,j,...,j'$.
    Suppose that for all $i' =1,...,2^k$ $T_{c,g}(Q)_{li'}=0$ implies $T_{c,g}(Q)_{hi'}=0$ (otherwise the $i'$-th column is not in the column space of $T_{c'}(Q)$, c.f. b2).
    By slightly abusing notation, we let $c_l=0$ be the correct answering parameter and $g_l\neq 0$ be the guessing parameter of the item in the $l$-th row of $T_{c,g}(Q)$.
    Let $\AA = \ee_{j'}$. Then, the $\AA$-th element of the $l$-th row is $g_l\neq 0$ (equivalently, $\AA$ is NOT able to answer that item).

    \bigskip

    \begin{itemize}
    \item[d1]Suppose that the $\AA$-th element
    of the $h$-th row of $T_{c,g}(Q)$ is non-zero. Let $\mathbf 0 = (0,...,0)$ be the (zero) attribute
    that has precisely one few attribute than $\AA$. Then, the $\mathbf 0$-th element of the $l$-th row of $T_{c,g}(Q)$ equals the $\AA$-th element of that row (being $g_l$). The $\mathbf 0$-th and the $\AA$-th elements of the $h$-th row of $T_{c,g}(Q)$ are different. This is because the $\mathbf 0$-th and the $\AA$-th elements of the $h$-th row are
$$\prod_{i'=1}^{j'} g_{i'},\qquad c_{j'}\prod_{i'=1}^{j'-1} g_{i'}.$$
    Thereby, we can identify the vector from either the $\AA$-th or the $\mathbf 0$-th column vector. (Note that the $\mathbf 0$-th column of $T_{c,g}(Q)$ is the first column, which is not a zero vector.)
    \item[d2] Suppose that the $\AA$-th element of the $h$-th row of $T_{c,g}(Q)$ is zero. Then, the $\AA$-th column is not in the column space of $T_{c'}(Q)$, because its $l$-th element is nonzero and the $h$-th element is zero (c.f. b2).
    \end{itemize}
    %
%
%
    \end{itemize}

\end{itemize}

\item[Case 2] Either the $h$-th or $l$-th row vector of $T_{c'}(Q')$ is a zero vector. Since both the $h$-th and $l$-th rows of
$T_{c,g}(Q)$ are nonzero vectors, we are always able to identify a column in $T_{c,g}(Q)$ that is not in the column space of $T_{c'}(Q^{\prime })$.
\end{enumerate}
\end{proof}

\bigskip

\begin{proof}[Proof of Proposition \ref{PropIncomp}]
\subsubsection*{Step 1}
We first identify two row vectors such that they are identical in $T(Q')$ but distinct in $T(Q)$. It turns out that
we only need to consider the first $k$ items. Consider $Q'$ such
that $Q'_{1:k}$ is incomplete. We discuss the following situations.
\begin{enumerate}
\item There are two row vectors, say the $i$-th and $j$-th row vectors ($1\leq i,j \leq
k$), in $Q_{1:k}^{\prime }$ that are identical. Equivalently, two
items require exactly the same attributes according to $Q'$. Then,
the row vectors in $T(Q^{\prime })$ corresponding to these two items
are identical. All of the first $2^k - 1$ row vectors in  $T(Q)$
must be different, because $T(Q)_{1:(2^{k}-1)}$ has rank $2^{k}-1$.

\item No two row vectors in $Q_{1:k}^{\prime }$ are
identical. Then, among the first $k$ rows of $Q^{\prime }$ there is
at least one row vector containing two or more non-zero entries.
That is, there exists $1\leq i\leq k$ such that
\begin{equation*}
\sum_{j=1}^{k}Q_{ij}^{\prime }>1.
\end{equation*}%
This is because if each of the first $k$ items requires only one
attribute and $Q_{1:k}^{\prime }$ is not complete, there are at
least two items that require the same attribute. Then, there are two
identical row vectors in $ Q_{1:k}^{\prime }$ and it belongs to the
first situation. We define
\begin{equation*}
a_{i}=\sum_{j=1}^{k}Q_{ij}^{\prime },
\end{equation*}%
the number of attributes required by item $i$ according to $Q'$.

Without loss of generality, assume $a_{i}>1$ for $i=1,...,n$ and
$a_{i}=1$ for $i=n+1,...,k$. Equivalently, among the first $k$
items, only the first $n$ items require more than one attribute
while the $(n+1)$-th through the $k$-th items require only one
attribute each, all of which are distinct. Without loss of
generality, we assume $Q_{ii}^{\prime }=1$ for $i=n+1,...,k$ and
$Q_{ij}=0$ for $i=n+1,...,k$ and $i\neq j$.

\begin{enumerate}
\item \label{a} $n=1$. Since $a_{1}>1$, there exists $i>1$ such that $%
Q'_{1i}=1$. Then, the row vector in $T(Q^{\prime })$ corresponding
to $I_{1}\wedge
I_{i}$ (say, the $l$-th row in $T(Q^{\prime })$) and the row vector of $%
T(Q^{\prime })$ corresponding to $I_{1}$ are identical. On the other
hand, the first row and the $l$-th row are different for $T(Q)$
because $T(Q)_{1:(2^{k}-1)}$
is a full-rank matrix. The above statement can be written as%
\begin{equation*}
B_{Q^{\prime }}(I_{1}\wedge I_{i})=B_{Q^{\prime }}(I_{1}),\quad
B_{Q}(I_{1}\wedge I_{i})\neq B_{Q}(I_{1}).
\end{equation*}

\item $n>1$ and there exists $j>n$ and $i\leq n$ such that $Q'_{ij}=1$. Then by
the same argument as in \eqref{a}, we can find two rows that are
identical in $
T(Q^{\prime })$ but different in $T(Q)$. In particular,%
\begin{equation*}
B_{Q^{\prime }}(I_{j}\wedge I_{i})=B_{Q^{\prime }}(I_{i}),\quad
B_{Q}(I_{j}\wedge I_{i})\neq B_{Q}(I_{i}).
\end{equation*}

\item $n>1$ and for each $j>n$ and $i\leq n$, $Q'_{ij}=0$. Let the $i^{\ast }$-th
row in $T(Q')$ correspond to $I_{1}\wedge ...\wedge I_{n}$. Let
the $i_{h}^{\ast }$-th row in $T(Q')$ correspond to $I_{1}\wedge
...\wedge I_{h-1}\wedge I_{h+1}\wedge ...\wedge I_{n}$ for
$h=1,...,n$.

We claim that there exists an $h$ such that the $i^{\ast }$-th row and the $%
i_{h}^{\ast }$-th row are identical in $T(Q^{\prime })$, that is%
\begin{equation}\label{n1}
B_{Q^{\prime }}(I_{1}\wedge ,...,\wedge I_{h-1}\wedge I_{h+1}\wedge
,...,\wedge I_{n})=B_{Q^{\prime }}(I_{1}\wedge ,...,\wedge I_{n}).
\end{equation}
We prove this claim by contradiction. Suppose that there does not
exist such an $h$. This is equivalent to saying that for each $j\leq
n$ there exists an $\alpha _{j}$ such that $Q'_{j\alpha _{j}}=1$ and
$Q'_{i\alpha _{j}}=0$ for all $1\leq i\leq n$ and $i\neq j$.
Equivalently, for each $j\leq n$, item $j$ requires at least one
attribute that is not required by other first $n$ items. Consider
$$\mathcal{C}_{i}=\{j:\mbox{there exists $i\leq i' \leq n$ such that $Q'_{i' j}=1$}\}.$$
Let $\#(\cdot )$ denote the cardinality of a set. Since for each
$i\leq n$ and $j>n$, $Q'_{ij}=0$, we have that
$\#(\mathcal{C}_{1})\leq n$. Note that $Q'_{1\alpha _{1}}=1$ and
$Q'_{i\alpha _{1}}=0$ for all $2\leq i\leq n$. Therefore, $\alpha
_{1}\in \mathcal{C}_{1}$ and $\alpha_1 \notin \mathcal{C}_{2}$.
Therefore, $\#(\mathcal{C}_{2})\leq n-1$. By a similar argument and
induction, we have that $a_{n}=\#(\mathcal{C}_{n})\leq 1 $. This
contradicts the fact that $a_{n}>1$. Therefore, there exists an $h$
such that \eqref{n1} is true. As for $T(Q)$, we have that
$$B_{Q}(I_{1}\wedge ,...,\wedge I_{h-1}\wedge I_{h+1}\wedge
,...,\wedge I_{n})\neq B_{Q}(I_{1}\wedge ,...,\wedge I_{n}).$$
\end{enumerate}

\end{enumerate}

\subsubsection*{Step 2}

For the situations 1, 2a, and 2b, the identification of the column vector is completely identical to that of the Proposition \ref{PropComp}. For those three situations, we essentially  identified one row corresponding to a single item and another row corresponding to a combination of single-attribute items. We need to provide additional proof for situation 2c, that is, the follow-up analysis whence \eqref{n1} is established. Without loss generality, we assume that
\begin{equation}
B_{Q^{\prime }}(I_{1}\wedge ,...,\wedge I_{n-1})=B_{Q^{\prime }}(I_{1}\wedge ,...,\wedge I_{n}), \quad B_{Q}(I_{1}\wedge ,...,\wedge I_{n-1})\neq B_{Q}(I_{1}\wedge ,...,\wedge I_{n}).
\end{equation}
Let $h$ be the row corresponding to $I_{1}\wedge ,...,\wedge I_{n-1}$ and $l$ be the row to $I_{1}\wedge ,...,\wedge I_{n}$.

\begin{itemize}
\item[a] Both the $h$-th and the $l$-th row of $T_{c'}(Q')$ are nonzero.
Among the first $2^n$ elements of $B_{c,g,Q}(I_{1}\wedge ,...,\wedge I_{n-1})$ there exists a nonzero element, say corresponding to attribute $\AA$. Let $\AA'$ be the attribute identical to $\AA$ for the first $n-1$ attributes and their $n$-th elements are different. Then, the $\AA$-th and $\AA'$-th elements of $B_{c,g,Q}(I_{1}\wedge ,...,\wedge I_{n-1})$ are identical (and nonzero).  The $\AA$-th and $\AA'$-th elements of $B_{c,g,Q}(I_{1}\wedge ,...,\wedge I_{n})$ must be different. This is because $\AA$-th and $\AA'$-th elements of $B_{c,g,Q}(I_{1}\wedge ,...,\wedge I_{n})$ are the products of the corresponding elements in $B_{c,g,Q}(I_{1}\wedge ,...,\wedge I_{n-1})$ with $c_l$ and $g_l$ respectively and $c_l \neq g_l$. Then, either the $\AA$-th or the $\AA'$-th column of $T_{c,g}(Q)$ is not in the column space of $T_{c'}(Q')$.


\item[b] Either the $h$-th or the $l$-th row of $T_{c'}(Q')$ is a zero vector. The identification of the column vector is straightforward.
\end{itemize}
%
\end{proof}

\bibliographystyle{chicago}
\bibliography{bibstat}

\begin{thebibliography}{}

\bibitem[\protect\citeauthoryear{Chiu, Douglas, and Li}{Chiu
  et~al.}{2009}]{Chiu}
Chiu, C., J.~Douglas, and X.~Li (2009).
\newblock Cluster analysis for cognitive diagnosis: Theory and applications.
\newblock {\em Psychometrika\/}~{\em 74\/}(4), 633--665.

\bibitem[\protect\citeauthoryear{de~la Torre}{de~la Torre}{2008}]{dela2}
de~la Torre, J. (2008).
\newblock An empirically-based method of q-matrix validation for the {DINA}
  model: Development and applications.
\newblock {\em Journal of Educational Measurement\/}~{\em 45}, 343--362.

\bibitem[\protect\citeauthoryear{de~la Torre and Douglas}{de~la Torre and
  Douglas}{2004}]{dela}
de~la Torre, J. and J.~Douglas (2004).
\newblock Higher order latent trait models for cognitive diagnosis.
\newblock {\em Psychometrika\/}~{\em 69}, 333--353.

\bibitem[\protect\citeauthoryear{DiBello, Stout, and Roussos}{DiBello
  et~al.}{1995}]{DiBello}
DiBello, L., W.~Stout, and L.~Roussos (1995).
\newblock Unified cognitive psychometric assessment likelihood-based
  classification techniques.
\newblock {\em Cognitively diagnostic assessment. Hillsdale, NJ: Erlbaum\/},
  361--390.

\bibitem[\protect\citeauthoryear{Gelman, Jakulin, Pittau, and Su}{Gelman
  et~al.}{2008}]{Gelman08}
Gelman, A., A.~Jakulin, M.~G. Pittau, and Y.-S. Su (2008).
\newblock A weakly informative default prior distribution for logistic and
  other regression models.
\newblock {\em Annals of Applied Statistics\/}~{\em 2\/}(4), 1360--1383.

\bibitem[\protect\citeauthoryear{Hartz}{Hartz}{2002}]{Hartz}
Hartz, S. (2002).
\newblock A bayesian framework for the unified model for assessing cognitive
  abilities: Blending theory with practicality.
\newblock {\em Doctoral Dissertation, University of Illinois,
  Urbana-Champaign\/}.

\bibitem[\protect\citeauthoryear{Henson and Templin}{Henson and
  Templin}{2005}]{Henson}
Henson, R. and J.~Templin (2005).
\newblock Hierarchical log-linear modeling of the skill joint distribution.
\newblock {\em External Diagnostic Research Group Technical Report\/}.

\bibitem[\protect\citeauthoryear{Junker and Sijtsma}{Junker and
  Sijtsma}{2001}]{Junker}
Junker, B. and K.~Sijtsma (2001).
\newblock Cognitive assessment models with few assumptions, and connections
  with nonparametric item response theory.
\newblock {\em Applied Psychological Measurement\/}~{\em 25}, 258--272.

\bibitem[\protect\citeauthoryear{Leighton, Gierl, and Hunka}{Leighton
  et~al.}{2004}]{AHM}
Leighton, J.~P., M.~J. Gierl, and S.~M. Hunka (2004).
\newblock The attribute hierarchy model for cognitive assessment: A variation
  on tatsuoka's rule-space approach.
\newblock {\em Journal of Educational Measurement\/}~{\em 41}, 205--237.

\bibitem[\protect\citeauthoryear{Liu, Xu, and Ying}{Liu et~al.}{2011}]{LXY2011}
Liu, J., G.~Xu, and Z.~Ying (2011).
\newblock Theory of self-learning $q$-matrix.
\newblock {\em Preprint, arXiv:1010.6120\/}.

\bibitem[\protect\citeauthoryear{Roussos, Templin, and Henson}{Roussos
  et~al.}{2007}]{RoussosTH}
Roussos, L.~A., J.~L. Templin, and R.~A. Henson (2007).
\newblock Skills diagnosis using {IRT}-based latent class models.
\newblock {\em Journal of Educational Measurement\/}~{\em 44}, 293--311.

\bibitem[\protect\citeauthoryear{Rupp}{Rupp}{2002}]{Rupp2002}
Rupp, A. (2002).
\newblock Feature selection for choosing and assembling measurement models: A
  building-block-based organization.
\newblock {\em Psychometrika\/}~{\em 2}, 311--360.

\bibitem[\protect\citeauthoryear{Rupp and Templin}{Rupp and
  Templin}{2008}]{Rupp20082}
Rupp, A. and J.~Templin (2008).
\newblock Effects of q-matrix misspecification on parameter estimates and
  misclassification rates in the dina model.
\newblock {\em Educational and Psychological Measurement\/}~{\em 68}, 78--98.

\bibitem[\protect\citeauthoryear{Rupp, Templin, and Henson}{Rupp
  et~al.}{2010}]{Rupp}
Rupp, A., J.~Templin, and R.~A. Henson (2010).
\newblock {\em Diagnostic Measurement: Theory, Methods, and Applications}.
\newblock Guilford Press.

\bibitem[\protect\citeauthoryear{Stout}{Stout}{2007}]{Stout2007}
Stout, W. (2007).
\newblock Skills diagnosis using {IRT}-based continuous latent trait models.
\newblock {\em Journal of Educational Measurement\/}~{\em 44}, 313--324.

\bibitem[\protect\citeauthoryear{Tatsuoka}{Tatsuoka}{2002}]{TatsuokaC}
Tatsuoka, C. (2002).
\newblock Data-analytic methods for latent partially ordered classification
  models.
\newblock {\em Applied Statistics (JRSS-C)\/}~{\em 51}, 337--350.

\bibitem[\protect\citeauthoryear{Tatsuoka}{Tatsuoka}{1983}]{Tatsuoka1983}
Tatsuoka, K. (1983).
\newblock Rule space: an approch for dealing with misconceptions based on item
  response theory.
\newblock {\em Journal of Educational Measurement\/}~{\em 20}, 345--354.

\bibitem[\protect\citeauthoryear{Tatsuoka}{Tatsuoka}{1985}]{Tatsuoka1985}
Tatsuoka, K. (1985).
\newblock A probabilistic model for diagnosing misconceptions in the pattern
  classification approach.
\newblock {\em Journal of Educational Statistics\/}~{\em 12}, 55--73.

\bibitem[\protect\citeauthoryear{Tatsuoka}{Tatsuoka}{2009}]{Tatsuoka}
Tatsuoka, K. (2009).
\newblock {\em Cognitive assessment: an introduction to the rule space method}.
\newblock CRC Press.

\bibitem[\protect\citeauthoryear{Templin}{Templin}{2006}]{Templin2006}
Templin, J. (2006).
\newblock {CDM}: cognitive diagnosis modeling with mplus [computer software].
  (available from http://jtemplin.myweb.uga.edu/cdm/cdm.html).

\bibitem[\protect\citeauthoryear{Templin and Henson}{Templin and
  Henson}{2006}]{Templin}
Templin, J. and R.~Henson (2006).
\newblock Measurement of psychological disorders using cognitive diagnosis
  models.
\newblock {\em Psychological Methods\/}~{\em 11}, 287--305.

\end{thebibliography}


\end{document}